\documentclass[a4paper,UKenglish]{article}

\usepackage{amsthm}
\usepackage{amsmath, amsfonts}
\usepackage{graphicx}
\usepackage{todonotes}
\usepackage{theoremref}
\usepackage{framed}

\DeclareMathOperator{\Tr}{Tr}

\newcommand{\ket}[1]{|#1\rangle}
\newcommand{\bra}[1]{\langle #1|}

\newcommand{\ketbra}[2]{|#1\rangle\langle #2|}

\newcommand{\NN}{\mathbf{N}}

\newcommand{\Ans}{\mathrm{Ans}}

\newcommand{\mybar}[1]{\lambda}

\newtheorem{theorem}{Theorem}
\newtheorem{claim}{Claim}
\newtheorem{lemma}{Lemma}
\newtheorem{definition}{Definition}

\newtheorem{corollary}{Corollary}

\newtheorem{remark}{Remark}

\newcommand{\COMMENT}[1]{}

\title{Privacy in Quantum Communication Complexity}
\date{}

\author{Iordanis Kerenidis\footnote{LIAFA, CNRS, Universit\'{e} Paris Diderot, France. jkeren@liafa.univ-paris-diderot.fr}
\hspace{0.5cm}
Mathieu Lauri{\`{e}}re\footnote{LIAFA, Universit\'{e} Paris Diderot. mathieu.lauriere@univ-paris-diderot.fr}
\hspace{0.5cm}
Fran{\c c}ois Le Gall \footnote{The University of Tokyo, Japan. legall@is.s.u-tokyo.ac.jp}
\hspace{0.5cm}
Mathys Rennela \footnote{University of Nijmegen, the Netherlands. mathys.rennela@gmail.com}}

\begin{document}

\maketitle

\begin{abstract}
In two-party quantum communication complexity, Alice and Bob receive
some classical inputs and wish to compute some function that depends
on both these inputs, while minimizing the communication. This model
has found numerous applications in many areas of computer science. One
question that has received a lot of attention recently is whether it
is possible to perform such protocols in a private way. We show that
defining privacy for quantum protocols is not so straightforward and
it depends on whether we assume that the registers where Alice and Bob
receive their classical inputs are in fact classical registers (and
hence unentangled with the rest of the protocol) or quantum registers
(and hence can be entangled with the rest of the protocol or the
environment). We provide new quantum protocols for the Inner Product
function and for Private Information Retrieval, and show that the
privacy assuming classical input registers can be exponentially
smaller than the privacy assuming quantum input registers. We also
argue that the right notion of privacy of a communication protocol is
the one assuming classical input registers, since otherwise the
players can deviate considerably from the protocol.
\end{abstract}

\section{Introduction}
In two-party communication complexity \cite{Y79}, Alice and Bob receive inputs $x$ and $y$ and wish to compute some function that depends on both these inputs, while minimizing the communication. This model has found numerous applications in many areas of computer science. One question that has received a lot of attention recently is whether it is possible to perform such protocols in a private way. 

In classical communication protocols, the privacy loss (or information cost) is defined as the information that the transcript reveals to each player about the input of the other one.
In this model, one is interested in the privacy loss of a specific protocol and hence we only consider the case where the players honestly follow the protocol and not how they can increase the information by deviating from the protocol.

In quantum communication protocols \cite{Y93}, Alice and Bob receive classical inputs $x$ and~$y$ and wish to compute a function $f(x,y)$. We consider three quantum registers $A,M,B$ that correspond to Alice's workspace, the message qubits, and Bob's workspace. 
At each round of the protocol, one player applies a unitary operation on his workspace and the message qubits, and sends the message qubits to the other player who continues the protocol. Since the message qubits can be reused throughout the protocol and copying of the quantum states may be impossible, we cannot define a transcript. Hence, we know of no way to define notions of privacy or quantum information cost, other than by a {\em round-by-round} definition. 

It is not hard to see, by a chain rule argument, that the classical definition of privacy loss is equivalent to a {\em round-by-round} definition where for every round~$k$, we calculate the information that the message at round~$k$ reveals about the sender's input to the receiver, who knows his input and has kept a copy of all previous messages in his workspace. 

Again, this definition is not readily applicable to quantum protocols, since the players may not be able to copy the messages and continue the protocol at the same time. Nevertheless, they have a quantum workspace, where, depending on the protocol, they may keep information about previous messages. We would like to calculate how much information every new message reveals to them, given that they already know their own input and have kept some information in their quantum workspace according to the protocol.

There is one additional important issue to consider. 
Each player has a register where the input is written in the beginning of the protocol. This input is always a classical input. 
One natural possibility is therefore to consider that the input register is a classical register, meaning it cannot be entangled with the workspaces and the message space. The second possibility is to consider that the input is written in a quantum register, which could be entangled with the players' workspaces or even with the environment. We discuss below in more details these two possibilities (formal definitions and more complete discussion is given in Section \ref{sec-defs}).

\emph{Privacy for quantum protocols with classical input registers.}
Privacy with classical input registers has previously been discussed for some classes of quantum protocols, for example in \cite{K02,LG12}. In particular, Klauck defined a notion of honest behaviour of the players, according to which, at every round of the protocol, the state of the message qubits sent must be equal to the one prescribed by the protocol.  He, then, considered the privacy of Disjointness by looking at a quantum protocol with pure state messages. However, when applied to protocols with mixed state messages, his definition allows for a player to change the execution of the protocol considerably. 
Hence, in this paper we propose a definition of honest execution of a protocol that takes into consideration protocols in which the players can send mixed states, and is equivalent to Klauck's definition for pure state protocols.
We then define Alice's and Bob's privacy loss for a protocol $\pi$, and denote these quantities by
$L_A(\pi)$ and $L_B(\pi)$, respectively.

\emph{Privacy for quantum protocols with quantum input registers.}
There have been definitions of quantum information cost where the input registers are considered to be quantum registers that can hold superpositions of inputs and be entangled with other registers (\cite{CS14,JRS03,BB14,T14}).
This is the case for example with the "superposed" information cost for a protocol $\pi$. We denote it (see Definition \ref{def-SIC}) by $SIC_A(\pi)$, where Bob is allowed to have in his input register a superposition of all his possible inputs (corresponding to the input distribution) instead of the fixed classical input~$y$ he received, and he is allowed to measure this register in the computational basis at any moment during the execution of~$\pi$. While this is certainly a strategy that Bob can follow in order to acquire more information about Alice's input, as we said before, we are not in a cryptographic scenario with cheating players, rather we want to compute the privacy loss of the specific protocol that computes the function. Also note that if Bob starts with a superposition of all possible inputs, Alice and Bob are not able to compute the value of the function $f$ on the received inputs $(x,y)$, but only on a point $(x,y')$ for a random $y'$. Bob's superposed information cost, denoted $SIC_B(\pi)$, is defined in a symmetric way. 
Recently, another definition of quantum information cost for quantum protocols, denoted in this paper by 
$QIC_A(\pi)$ and $QIC_B(\pi)$, was proposed by Touchette \cite{T14}. This definition has nice properties, e.g., it is equal to the amortized quantum communication complexity \cite{T14}. In this case, the input registers are initially entangled with an external register (an environment, not accessible to the players). Note that the quantum information cost does not compute the information a message reveals about a player's input (as in the usual notion of privacy) but about the environment register.

\paragraph*{Our results}
The main goal of this paper is to investigate quantum communication complexity under these different variants of privacy, 
and in particular study the differences between privacy with classical input registers and privacy with quantum input registers.
In the present work, we first prove the following inequalities between these definitions.

\noindent
{\bf Result 1}:
$L_A(\pi) \leq SIC_A(\pi) \leq QIC_A(\pi) \mbox{ and } L_B(\pi) \leq SIC_B(\pi) \leq QIC_B(\pi).$

We then show that, in some cases, the gaps can be exponentially large. This is done by considering the privacy of fundamental tasks such as Private Information Retrieval and the Inner Product function. In order to obtain these gaps, one of the main contributions of this paper is the construction of protocols for these tasks with small privacy loss (i.e., in the setting where the input registers are classical).  

We believe that constructing quantum protocols with small privacy loss gives new insight about the relation between privacy and quantum information. 
In particular, we believe that the notions of superposed and quantum information cost are much more suited as tools in order to lower bound the communication complexity than as tools to measure the privacy of the communication protocol under consideration. Note that for $SIC$, the parties can indeed, as mentioned above, considerably deviate from the protocol and 
may not even compute the function $f$ on the received inputs; for $QIC$ we do not measure the information revealed about each player's input but about the register $R$. In comparison, privacy loss with classical input registers appears to be more suited as a tool to discuss the privacy of a specific protocol computing a function with classical inputs. 
Notice also that in the classical case, if we allow Alice to run the protocol with a random input instead of the input $x$ she received, then private protocols, like for the IdMinimum function (where the output is $\min(x,y)$ together with the identity of the player who has this value) are rendered not private.

We describe below in more details our results for the Inner Product function and for Private Information Retrieval.

\emph{The privacy of Inner Product.}
We prove a simple gap between quantum communication complexity and privacy loss.

\noindent
{\bf Result 2}:
There exists a quantum protocol for Inner Product, which is perfectly private for Bob and where Alice's privacy loss is only $n/2+1/2$.

We also show that for the protocol we contruct the superposed and quantum information cost is basically $n/2$ for both parties, hence providing a gap between these notions. 

\emph{The privacy of Private Information Retrieval.}
Private Information Retrieval has been extensively studied so as to find the minimum communication necessary between the user and one or more servers, while keeping the perfect privacy of the user. Here we consider the one-server setting: the server has for input a database $x=x_1\cdots x_n\in\{0,1\}^n$, the user has for input an index $i\in \{1,\ldots,n\}$, and the goal is for the user to output $x_i$. It is well known that any classical protocol perfectly private for the user (i.e., in which the server obtains no information about~$i$) requires $\Omega(n)$ bits of communication \cite{CGKS95}. Moreover, the quantum communication complexity, as well as the superposed and quantum information costs are also $\Omega(n)$ \cite{JRS09}.

Recently, Le Gall \cite{LG12} showed that there exists a quantum protocol for this task, perfectly private for the user (according to Definition~\ref{def-priv}), with communication complexity $O(\sqrt{n})$. This upper bound has then be improved to $O(n^{1/3})$ by Ruben Brokkelkamp \cite{RB13}.
Here we ask the question: Can these upper bounds be further improved? Or, more generally, how much information does a single server have to leak about the database in any protocol which is perfectly private for the user?
We show the following surprising result.

\noindent
{\bf Result 3}:
There exists a quantum protocol for Private Information Retrieval, which is perfectly private for the user and in which the server's privacy loss is polylogarithmic on the size of the database. Moreover its communication complexity is also polylogarithmic on the size of the database.

This provides the first exponential separation between the different notions of privacy of quantum protocols (namely, privacy loss versus superposed and quantum information costs).
 
The proof has two steps: first, we show how to take any $\ell$-server classical PIR scheme and translate it into a quantum one-server scheme, such that the index remains perfectly private. 
Then, we use a classical PIR scheme with a logarithmic number of servers and polylogarithmic communication \cite{CGKS95}, which implies that the privacy loss about the database is polylogarithmic, since it is always less than the communication. 

Finally, we improve the above upper bounds when the user and the server share prior entanglement: we construct a new quantum protocol for Private Information Retrieval, perfectly private for the user, where the server's privacy loss is $O(\log n)$ bits. The communication complexity of this protocol is $O(\log n)$, which is optimal since, even with prior entanglement, the quantum communication complexity of the Index Function is $\Omega(\log n)$.

\section{Preliminaries}
In this paper we write, for a positive integer $p$, $[p] := \{1, 2, \dots, p\}$ and, 
for two positive integer $p<q$, write $[p,q] := \{p,p+1,\dots,q\}$.

In two-party communication complexity, Alice and Bob receive inputs $x$ and $y$ respectively and wish to compute some function $f(x,y)$ that depends on both these inputs, while minimizing the {\it communication cost}, i.e., the number of exchanged bits. The {\it communication complexity} of a function is the least amount of communication possible in a protocol computing $f$. We refer to \cite{KN97} for details about classical communication complexity, and to \cite{BNSW98} for asymmetric communication complexity.

In two-party quantum communication complexity, the players are now allowed to exchange quantum bits. The standard model consists of three quantum registers: A, M and~B. Here A and~B are private workspaces of Alice and Bob respectively, while M is used to communicate qubits and is sent from one player to the other one. Additionally, Alice and Bob hold a register (classical or quantum), say X and Y respectively, where they store their respective input. At every round, one player applies a unitary operation on their workspace and the message qubits (that also depends on their input) and sends the message qubits to the other player who continues the protocol. 

In the above setting, the Inner Product (IP) problem consists in computing $f(x,y) = x\cdot y := \sum_i x_i \cdot y_i$. In \cite{CG88}, it is proved as a particular case of the bounded error setting, that computing classically and perfectly IP requires a communication of $n$, and the same holds for quantum protocols \cite{CDNT99}. 

Another well studied problem is Private Information Retrieval (PIR): a user, whose input is an index $i \in [n]$, interacts with a server holding a database $x = (x_j)_{j\in[n]} \in \{0,1\}^n$. The goal for the user is to learn $x_i$ in such way that the server does not learn his index,~$i$. In \cite{CGKS95P,CGKS95}, it is shown that the communication complexity of this problem is $\Omega(n)$.
The same paper also shows that it is possible to improve the communication complexity if the user can interact with several independent servers: in this setting it is possible to obtain a communication polylogarithmic in $n$. In \cite{LG12}, the author gives a quantum protocol using a single server and only $O(\sqrt n)$ qubits of communication, which yields a quadratic improvement over what is possible classically. This upper bound has then be improved to $O(n^{1/3})$ by Ruben Brokkelkamp \cite{RB13}. In both cases the protocol is perfectly private for the user (as long as the server follows exactly the prescribed scheme). If we allow the players to create superpositions of inputs or act as specious adversaries, then it is known that the communication from the server must be linear \cite{JRS09, BB14}.

The privacy will be analyzed with information theoretical tools. More precisely, $S(X)$ will denote the entropy of $X$, that is $S(X)$ is equal to $- \sum_x p_x \log(p_x)$ if $X$ is a classical random variable taking value $x$ with probability $p_x$, or to $- \Tr(\rho_X \log (\rho_X))$ if $X$ is quantum register whose state is denoted by $\rho_X$. If $A,B$ and $C$ are either classical random variables or quantum registers, the mutual information between $A$ and $B$ (resp. the mutual information between $A$ and $B$ conditioned on $C$) is defined by $I(A:B) = S(A)+S(B)-S(AB)$ (resp. $I(A:B|C) = I(AC:B) - I(C:B) = S(AC)+S(BC) - S(C)-S(ABC)$, which can be interpreted as the knowledge that $A$ gives about $B$ provided that we already knew $C$).

\section{Definitions of privacy for quantum protocols and their relation}\label{sec-defs}

In classical communication protocols, the privacy loss is defined as the information that the transcript of the communication reveals to each player about the input of the other one. Using a chain rule argument, it is not hard to see that the classical definition of privacy loss is equivalent to a {\em round-by-round} definition where for every round $k$, we calculate the information that the message at round $k$ reveals about each player's input to the other player, who already knows his input and has kept a copy of all previous messages in his workspace:
\begin{align*}
	I(\Pi : X | Y ) &= \sum_{k: \text{ odd}} I(M_{k}:X|Y,M_1,\ldots,M_{k-1}),\\
	I(\Pi : Y | X) &= \sum_{k :\text{ even}} I(M_{k}:Y|X,M_1,\ldots,M_{k-1}). 
\end{align*}

Note that we define the privacy loss of each player separately, since their input sizes or their privacy considerations can be different. For example, for Private Information Retrieval, we will look at protocols which are perfectly secure for the user (whose input has size $\log n$) and leaks a logarithmic amount of information about the database (whose size is $n$). Moreover, we do not condition on the value of the function that each player computes, since as we will see in the quantum setting, the protocol may not compute the function.

In quantum communication protocols, since there is no notion of transcript, we define notions of privacy or quantum information cost by a {\em round-by-round} definition. As we said, we will also differentiate between the case where the input registers are classical or quantum.

\subsection{Privacy for quantum protocols with a classical input register}

Let us first assume that the input registers of the two players are classical.
We start by describing Klauck's definition of an honest quantum protocol.
\begin{definition}[\cite{K02}] 
	A protocol is honest if both players are honest. A player is honest if for all rounds of the protocol, for all inputs he may have and for all sequences of pure state messages he may have received in the previous rounds, the density matrix of the message in the next round equals the density matrix defined by the protocol and the input. The behavior of the player on mixed states is defined by his behavior on pure state messages.
\end{definition}

Klauck used this definition for a protocol where all messages were pure states. Nevertheless, in a run of a general protocol the player might actually receive mixed state messages. 

Let us consider the following simple $2$-round scheme where the players use only two registers, $(Q,R)$:
\begin{enumerate}
	\item Alice prepares $\ket\phi = \frac 1 2 (\ket0_Q \ket 0_R + \ket 1_Q \ket 1_R)$ and sends register $R$ to Bob.
	\item Bob sends back register $R$ (without doing anything).
\end{enumerate}
In this situation, we would expect that Bob should not be able to save a copy of register~$R$ before returning it, since in that case, while according to the protocol Alice should have had the pure state $\ket\phi$ at the end of the two rounds, when Bob copies the register $R$, she ends up with an equal mixture of the states $\ket0 \ket0$ and $\ket1 \ket1$. Nevertheless, according to Klauck's definition, Bob's behaviour is permissible. Indeed, Bob receives a uniform mixture of the pure states $\ket 0$ and $\ket 1$ and for each one of the pure states, Bob can in fact copy the state and then return the register $R$ to Alice. While for each pure state, his behaviour is according to the protocol, in fact his overall behaviour is not!

We hence propose a different definition, more adept for protocols with mixed state messages, which coincides with Klauck's definition when messages are pure.
\begin{definition}
Let $\pi$ be a quantum protocol where, at each round $k$, the quantum registers corresponding to the message sent, Alice's workspace and Bob's workspace are denoted $M_k, A_k$ and $B_k$, respectively.
An honest execution of the protocol is such that for all $k$, the joint state in the registers  $A_k, M_k, B_k$ is equal to the state described by the protocol, up to a possible local operation on $A_k$ and a possible local operation on $B_k$.
\end{definition} 
We can now provide the definition of privacy loss 
\begin{definition}\label{def-priv}
For a protocol $\pi$, the privacy loss of Alice and Bob are defined as
\[
	L_A(\pi) = \sum_{k\,:\,\text{odd}} I(M_k:X|Y, B_k) \mbox{ and }   L_B(\pi) = \sum_{k\,:\,\text{even}} I(M_k:Y|X, A_k),
\]
where $X,Y$ are classical registers that hold the inputs according to the input distribution, and $M_k,A_k,B_k$ are quantum registers that correspond to the message qubits and Alice's and Bob's workspaces at round $k$. 
\end{definition}
It is easy to see that the privacy loss for any honest execution of the protocol is the same, hence we only need to consider the states described by the protocol itself.
Also, if according to the protocol Alice holds a pure state at some round, then an honest Bob can not entangle his workspace with Alice's state. Last if $\pi$ computes some function~$f$, then any honest execution also computes $f$. This is important, since as in the classical case, it makes sense to consider only the privacy of protocols that actually compute the function $f$.

\subsection{Privacy for quantum protocols with a quantum input register}

In the previous definition we considered $X$ and $Y$ to be classical registers. There have been definitions of quantum information cost where the registers $X$ and $Y$ are considered to be quantum registers that can hold superpositions of inputs and be entangled with other registers (e.g., \cite{CS14,JRS03,BB14,T14}).

For example we can introduce the "superposed" information cost for a protocol $\pi$. The definition we provide here is somewhat different from the one of \cite{CS14}, in order to make the comparison between the different notions of privacy more direct. In spirit, the idea of the "superposed" information cost is to allow one player to run the protocol with a superposition of inputs instead of the classical input she or he has received according to the input distribution.

\begin{definition}\label{def-SIC}
For a protocol $\pi$, the superposed information cost of Alice and Bob are defined respectively by
\[
	SIC_A(\pi) = \sum_{k\,:\,\text{odd}} I(M_k:X|Y, B_k) \mbox{ and }   SIC_B(\pi) = \sum_{k\,:\,\text{even}} I(M_k:Y|X, A_k),
\]
where for Alice's privacy loss, Alice follows the protocol $\pi$ with her classical input in register $X$ and Bob creates a superposition of his inputs in register $Y$ which he can measure in the computational basis at any round\,; and similarly for Bob.
\end{definition}

Note that, in the case of a product input distribution, this is certainly a strategy that Bob could follow in order to acquire more information about Alice's input, but as we said before, we are not in a cryptographic scenario with cheating players. Also note that if Bob has a superposition of all possible inputs, Alice and Bob are not able to compute the value of the function $f$ on the received inputs $(x,y)$, but only on a point $(x,y')$ for a random $y'$. 

Recently, another definition of quantum information cost for protocols with entanglement was proposed by Touchette \cite{T14}. This definition has very nice properties, for example, it is equal to the amortized quantum communication complexity. In this case, the registers $X$ and~$Y$ are initially entangled with a register $R$ (an environment, not accessible to the players) so that the initial input state of the players for a distribution $\mu$ is
$
\sum_{x,y} \mu(x,y) \ket {x,y}_R \ket x_X \ket y_Y.
$
Even though the registers $X$ and $Y$ contain a mixture of classical inputs, nevertheless these registers are entangled with the register $R$, which is in fact what appears in the definition of Touchette's information cost. In other words, we are counting the information the message reveals about the register $R$ and not the player's input register. For simplicity, we provide a variant of his definition for protocols without prior entanglement.
\begin{definition}\label{def-QIC}
For a protocol $\pi$, the Quantum Information Cost of Alice and Bob are defined as
\[
QIC_A(\pi)=\sum_{k\,:\,\text{odd}} I(M_k:R|Y,B_k) \mbox{ and } QIC_B(\pi)=\sum_{k\,:\,\text{even}} I(M_k:R|X,A_k),
\]
where the register $R$ holds a purification of the registers $X,Y$ and initially the registers $A$ and $B$ are equal to $\ket{0}$.
\end{definition}

\subsection{Relation between the different definitions of privacy}

We have seen three different definitions which measure in some way the information transmitted during the protocol. We believe that the notions of superposed and quantum information cost are much more suited as tools in order to lower bound the communication complexity than as tools to measure the privacy of the communication protocol under consideration (the parties can indeed, as mentioned above, considerably deviate from the protocol and 
may not even compute the function $f$ on the received inputs). In comparison, privacy with a classical register appears to be more suited to discuss the privacy of a protocol computing a function with classical inputs; nevertheless it is a weaker lower bound for communication. 

Notice also that in the classical case, if we allow Alice to run the protocol with a random input instead of the $x$ she received, then private protocols can be rendered not private. For example, the function IdMinimum, where Alice and Bob must output the identity of the party that holds the minimum of two numbers and its value, can be computed perfectly privately.
Indeed, the protocol where at each round $k=1\dots 2^n$ the players stop if one of them has input $k$, reveals no more than the output of the IdMinimum function. However, when Alice picks a random input and runs the protocol, the protocol is no longer private for Bob since he reveals his input with probability $1/2$.

We now prove a general inequality between these three notions of privacy
\begin{theorem}
For any protocol $\pi$ we have
\[
L_A(\pi) \leq SIC_A(\pi) \leq QIC_A(\pi) \mbox{ and } L_B(\pi) \leq SIC_B(\pi) \leq QIC_B(\pi).
\]
\end{theorem}

\begin{proof}
The inequalities $L(\pi) \leq SIC(\pi)$ are by definition since for the superposed information cost, we allow the players to measure their input at any point during the protocol, and hence they can do it at the beginning and run the protocol with a classical input. 

For $SIC(\pi) \leq QIC(\pi)$, we define the following state for any round of any protocol
\begin{align*}
	\frac{1}{2^n} \sum_{x,y} \ket{x}_{X_R} \ket{y}_{Y_R} \ket{x}_X \ket{y}_Y \ket{\phi_{xy}}_{AMB}.
\end{align*}
This state corresponds to the joint state when considering $QIC$, where $X_RY_R=R$ is the environment's register. It is also a purification of the state when considering, for example $SIC_A$, where now, $X_R$ is the register of the environment, while Bob has both registers $Y_RY$. We can, indeed, assume that Bob creates a superposition of his inputs in register $Y$ and appends an extra register $Y_R$, where he copies each classical input via a CNOT operation. Let us assume that Bob does not measure until the end of the protocol. Then, at round $k$,
\begin{eqnarray*}
 I(M_k:X|YY_RB_k) & = & I(M_k : XY_R|YB_k)- I(M_k : Y_R | YB_k) \\
& \leq & I(M_k : XY_R | YB_k) = I(M_k : X_RY_R | YB_k) 
\end{eqnarray*}
Summing over odd $k$ we obtain $SIC_A(\pi) \leq QIC_A(\pi)$.

Now, imagine Bob measures after round $\ell$, which is equivalent to tracing out $Y_R$ (or giving it back to the environment). We can prove as above that 
$\sum_{k=1}^\ell I(M_k:X|Y, B_k) \leq  \sum_{k=1}^\ell I(M_k:R|Y,B_k)$
 and after the measurement, the state is the same as the one in the privacy loss, which for any round $k'>\ell$ is smaller that the quantum information cost: 
\[
I(M_{k'} : X |Y B_{k'}) = I(M_{k'} : XY_R |Y B_{k'}) - I(M_{k'} : Y_R | XYB_{k'}) \leq I(M_{k'} : X_RY_R |Y B_{k'}).
\]
\end{proof}

\section{Privacy for Inner Product}
In this section we describe a quantum protocol for Inner Product and compute all different privacy quantities for it. 

The protocol is given in Fig.~\ref{fig:algorithmIP}. Here we assume that only Alice needs to learn the value of the function (then she could communicate it to Bob, leaking at most one bit of information about her input). 

\begin{figure}[ht!]
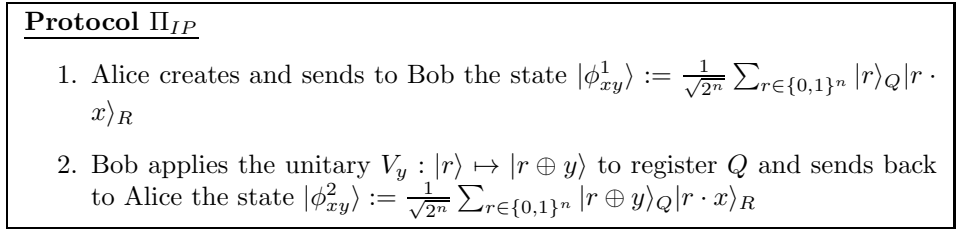

\begin{center}
\fbox{
\begin{minipage}{12 cm}
\underline{\textbf{Protocol $\Pi_{IP}$}}
\begin{enumerate}
	\item Alice creates and sends to Bob the state
	$	\ket {\phi_{xy}^1} :=  \frac 1 {\sqrt{2^n}} \sum_{r \in \{0,1\}^n} \ket r _Q \ket {r\cdot x}_R
	$
	\item Bob applies the unitary $V_y : \ket r \mapsto \ket{r\oplus y}$ to register $Q$ and sends back to Alice the state
	$	\ket {\phi_{xy}^2} := \frac 1 {\sqrt{2^n}} \sum_{r \in \{0,1\}^n} \ket {r \oplus y} _Q \ket {r\cdot x}_R
	$
\end{enumerate}
\end{minipage}
}
\end{center}
\caption{Quantum protocol for inner product.}\label{fig:algorithmIP}
\end{figure}

Let us prove the correctness of the protocol. Observe that
	\begin{align*}
		\ket {\phi_{xy}^2} =  \frac 1 {\sqrt{2^n}} \sum_{r \in \{0,1\}^n} \ket {r \oplus y} _Q \ket {r\cdot x}_R
					= \frac 1 {\sqrt{2^n}} \sum_{r \in \{0,1\}^n} \ket {r} _Q \ket {(r \oplus y)\cdot x}_R.
	\end{align*}
At the end of the protocol, Alice, by applying the unitary $U_x: \ket r \ket b \mapsto \ket r \ket{b\oplus r\cdot x}$, can transform $\ket {\phi_{xy}^2}$ to the state
	\begin{align*}
		\qquad \ket {\phi_{xy}^3} := \frac 1 {\sqrt{2^n}} \sum_{r \in \{0,1\}^n} \ket {r} _Q \ket {((r \oplus y)\cdot x) \oplus (r \cdot x)}_R = \left(\frac 1 {\sqrt{2^n}} \sum_{r \in \{0,1\}^n} \ket {r} _Q \right) \ket {x \cdot y}_R.
	\end{align*}
By measuring Register $R$, Alice obtains the bit $x \cdot y$.

Assuming the inputs are distributed uniformly, we then evaluate the privacy of this protocol.
\begin{theorem}\label{IP}
For the above protocol $\Pi_{IP}$ under uniform distribution of inputs, we have 
	\begin{eqnarray*}
		L_A(\Pi_{IP})  = n/2+1/2 \quad &,& \quad L_B(\Pi_{IP}) = 1. \\
		SIC_A(\Pi_{IP}) = n/2 + 1/2\quad &, &\quad SIC_B(\Pi_{IP}) = n/2 + 1/2. \\
		QIC_A(\Pi_{IP}) = n/2+1/2\quad &, & \quad QIC_B(\Pi_{IP}) = n/2+3/2. 
	\end{eqnarray*}
\end{theorem}
The proof of Theorem \ref{IP} is given below.
Note that, since Alice must output $x\cdot y$, the quantity $L_B$ is a least one for any protocol computing Inner Product,
which means that our protocol is optimal with respect to this quantity.
Also note that the lower bound of Cleve et al. \cite{CDNT99} on the quantum communication complexity of Inner Product 
shows that the sum of the privacy loss or information cost of both players is at least $n/2$.

\begin{proof}[Proof of Theorem \ref{IP}]
The proof is split into two claims, which we prove one by one. 
\begin{claim}\label{claimBob}
	Bob gets $n/2 +1/2$ bits of information from the first message and Alice one bit form the second message. More precisely: 
		$L_A(\Pi_{IP})=n/2+1/2$ and $L_B(\Pi_{IP})=1$.
\end{claim}
\begin{proof}
After receiving the first message, the information that Bob has about Alice's input is, by definition:
	\begin{align*}
		I(M_1:X|Y) = S(M_1|Y) -S(Y) - S(XYM_1)+S(XY) = S(M_1),
	\end{align*}
since $S(XYM_1)=S(XY)=2n$ and $M_1$ is independent of $Y$. It remains to calculate $S(M_1)$. Define
	\begin{align*}
		M_1^x = &\ketbra{\phi_{xy}^1}{\phi_{xy}^1} = \frac 1 {2^n} \sum_{r,r'} \ket r \ket {r\cdot x} \bra{r'} \bra{r'\cdot x}. 
	\end{align*}	
	Then 
	\begin{align*}
	\qquad M_1 = &\sum_{x \in \{0,1\}^n} \frac 1 {2^n} M_1^x = \frac 1 {2^{2n}} \sum_{r, r'} c(r,r',i,j) \ket r \ket i \bra {r'} \bra j,
	\end{align*}
	where the coefficient $c(r,r',i,j)$ is defined on $\{0,1\}^n \times \{0,1\}^n$ for $i,j\in \{0,1\}$ as:
	\begin{equation*}
		c(r,r',i,j) := \#\Big\{x \in \{0,1\}^n: \, r \cdot x = i, r' \cdot x= j\Big\} =
			\begin{cases}
				&2^n \quad \quad \hbox{if } r=r'=i=j=0		\\		
				&0 \quad \quad \hbox{  if } r = r', i \neq j \\
				& \qquad \quad \hbox{or } r=0,i=1 \hbox{ or } r'=0,j=1\\
				&2^{n-1} \quad \hbox{if } r = r'\neq 0, i=j \\
				& \qquad \quad \hbox{or } r=0\neq r',i=0 \\
				& \qquad \quad \hbox{or } r'=0\neq r,j=0\\
				&2^{n-2} \quad \hbox{otherwise}.
			\end{cases}
	\end{equation*}
	We can show by computing the matrix and its eigenvalues that $S(M_1) = n/2 + 1/2$ (up to exponentially small terms).

	Alice receives only one message from Bob, and after this message she has the state $\rho^2_{x,y} = \frac{1}{2^{2n}} \sum_{x,y} \ketbra{x}{x} \otimes \ketbra{\phi_{xy}^2}{\phi_{xy}^2}$\, with
	\begin{align*}
		\qquad \ket {\phi_{xy}^2} = \frac 1 {\sqrt{2^n}} \sum_{r \in \{0,1\}^n} \ket {r} \ket {(r\oplus y) \cdot x}.
	\end{align*} 
We have
\begin{eqnarray*}
I(M_2:Y|X) & = & S(M_2X)-S(X)-S(M_2XY)+S(SY)\\
& = & (n+1)-n+2n-2n \\
& = & 1,
\end{eqnarray*}
where we used the fact that the state  in the registers $M_2X$ has the same entropy as the following state (since there is a unitary on $M_2X$ that turns one into the other):
$\rho^3_{x,y} = \frac{1}{2^{2n}} \sum_{x,y} \ketbra{x}{x} \otimes \ketbra{\phi_{xy}^3}{\phi_{xy}^3}$ with
\[
\ket {\phi_{xy}^3} = \left(  \frac 1 {\sqrt{2^n}} \sum_{r \in \{0,1\}^n} \ket {r} \right) \ket {x \cdot y}.
\]

\end{proof}

\begin{claim}
For the Superposed and Quantum Information Cost of the protocol, we have
	\begin{eqnarray*}
		SIC_A(\Pi_{IP}) = n/2 + 1/2\quad &, &\quad SIC_B(\Pi_{IP}) = n/2 + 1/2. \\
		QIC_A(\Pi_{IP}) = n/2+1/2\quad &, & \quad QIC_B(\Pi_{IP}) = n/2+3/2. 
	\end{eqnarray*}
\end{claim}

\begin{proof}

We start with $SIC_A(\Pi_{IP})$. In fact, this is equal to $L_A$ since the information Alice leaks in the first message does not depend on whether Bob has a classical input or a superposition of inputs.

Let us compute $SIC_B(\Pi_{IP})$: Alice uses a uniform superposition of her inputs in register~$X$. 
In the beginning, she creates
\begin{equation*}
	\frac 1 {2^n} \sum_x \ket x \sum_{r \in \{0,1\}^n} \ket r  \ket {r\cdot x}.
\end{equation*}
Since Bob follows the protocol with a classical input $y$, Alice holds after Round 2 the state $\rho^2 = \frac{1}{2^n} \sum_y \ketbra{\phi^2_y}{\phi^2_y}$, with 
$\ket{\phi^2_y} = \frac 1 {2^n} \sum_x \ket x \sum_{r \in \{0,1\}^n} \ket r  \ket {(r\oplus y)\cdot x}$. 
Then, we can calculate
\begin{eqnarray*}
I(M_2:Y|X) & = & S(M_2X)-S(X)-S(M_2XY)+S(XY) \\
& = & (n/2+1/2)-(n/2+1/2)-n+(n+n/2+1/2) \\
& = & n/2+1/2.
\end{eqnarray*}

Now let us compute $QIC_A(\Pi_{IP})$.
For the first round, the state is
\begin{equation*}
	\ket{\phi^1} = \frac 1 {2^{3n/2}} \sum_{xy} \ket x_X \ket y_Y \ket {xy}_R\sum_{r \in \{0,1\}^n} \ket r  \ket {r\cdot x}.
\end{equation*}
Then, we have
\begin{align*}
	I(M_1:R|Y) &= S(M_1Y)-S(Y) + S(YR) - S(YM_1R) = S(M_1) = n/2 + 1/2.
\end{align*}
We used here the fact that $S(YM_1R) = S(YR)=n$, the fact that $M_1$ is independent of $Y$, and the equality $S(M_1)=n/2 + 1/2$ we have already proven when analyzing the privacy loss in the proof of Claim \ref{claimBob}. 

We finally compute $QIC_B(\Pi_{IP})$.
For the second round, we have
\begin{equation*}
	\ket{\phi^2} = \frac 1 {2^{3n/2}} \sum_{xy} \ket x_X \ket y_Y \ket {xy}_R \sum_{r \in \{0,1\}^n} \ket {r}  \ket {(r \oplus y) \cdot x},
\end{equation*}
and thus
\begin{eqnarray*}
	I(M_2:R|X) &= & S(M_2X)-S(X)-S(XM_2R) +S(XR) \\
& = & (n+1)-n-n+(n+n/2+1/2) \\
& = & n/2+3/2.
\end{eqnarray*}
\end{proof}
This concludes the proof of the second claim and hence of Theorem \ref{IP}.
\end{proof}

\begin{remark}
We can provide a tradeoff between the privacy loss of Alice and Bob in the following way:
Alice and Bob use their shared coins to pick a random $t \in [n]$. They apply Protocol $\Pi_{IP}$ for the first $t$ bits of their inputs. Then, for the remaining $n-t$ bits they switch roles and Bob sends the outcome to Alice. This new protocol is correct, since the inner product of $x,y$ is the XOR of the inner products of the smaller strings. Alice leaks at most $t/2+1/2$ bits from the first invocation and $1$ from the second one. Bob leaks at most~$1$ bit from the first one and $(n-t)/2+3/2$ from the second one and hence $n/2+4$ in total.
Note that allowing players to use public coins does not change the information cost or the privacy loss of the protocols.
\end{remark}	
\section{The privacy of Private Information Retrieval}

In this section we construct a quantum protocol for Private Information Retrieval with polylogarithmic privacy loss and polylogarithmic communication complexity, by describing a general method to convert a classical scheme for Private Information Retrieval with $\ell>1$ servers into a quantum scheme with a single server.

\paragraph*{Simulation of an $\boldmath{\ell}$-server classical scheme by a 1-server quantum scheme}
Consider a two-round classical scheme $\Pi_{PIR}$, where a user interacts with $\ell>1$ servers that each possess a copy of the database and are not allowed to interact with each other. We can describe such a scheme as in Fig.~\ref{fig:algorithmCPIR}, 
where $m_q, m_a,R \in \NN$. We assume that the distribution of queries that each server receives is uniform, and hence do not reveal any information about the user's input. This assumption is true for essentially all known classical protocols for (information-theoretic) Private Information Retrieval, including the protocols described in \cite{CGKS95P,CGKS95} that we will later use.

\begin{figure}[ht!]
\begin{center}
\fbox{
\begin{minipage}{12 cm}
\underline{\textbf{Protocol $\Pi_{PIR}$}}
\begin{enumerate}
	\item The user picks uniformly at random $r \in [R]$ that corresponds to a $\ell$-tuple of queries $\{ q_1^r,...,q_\ell^r\}$ and asks query $q_k^r \in \{0,1\}^{m_q}$ to server $i\in[\ell]$. 
	\item Each server $i$, who received $q_i^r$, sends his answer $a_i^r \in \{0,1\}^{m_a}$, to the user.
\end{enumerate}
\end{minipage}
}
\end{center}
\caption{General form of a 2-round $\ell$-server classical protocol for Private Information Retrieval.}\label{fig:algorithmCPIR}
\end{figure}

Let us now describe a quantum protocol $Q_{PIR}$ that simulates the classical protocol $\Pi_{PIR}$, but with a single server. The server and the user use $\ell$ query registers $Q_1, \dots, Q_\ell$ of size $m_q$ each and~$\ell$ answer registers $\Ans_1, \dots, \Ans_\ell$ of size $m_a$ each. Moreover the user also holds a private register $Q$ of size $\ell \cdot m_q$ to keep a copy of the queries.
The protocol is given in Fig.~\ref{fig:algorithmPIR}, where we use the notations $\ket {a_{[i-1]}^r}_{\Ans_{[i-1]}} := \ket {a_{1}^r}_{\Ans_1} \dots \ket {a_{i-1}^r}_{\Ans_{i-1}}$ and
$\ket 0_{\Ans_{[i,\ell]}} := \ket 0_{\Ans_i} \dots \ket 0_{\Ans_\ell}$.

\begin{figure}[ht!]
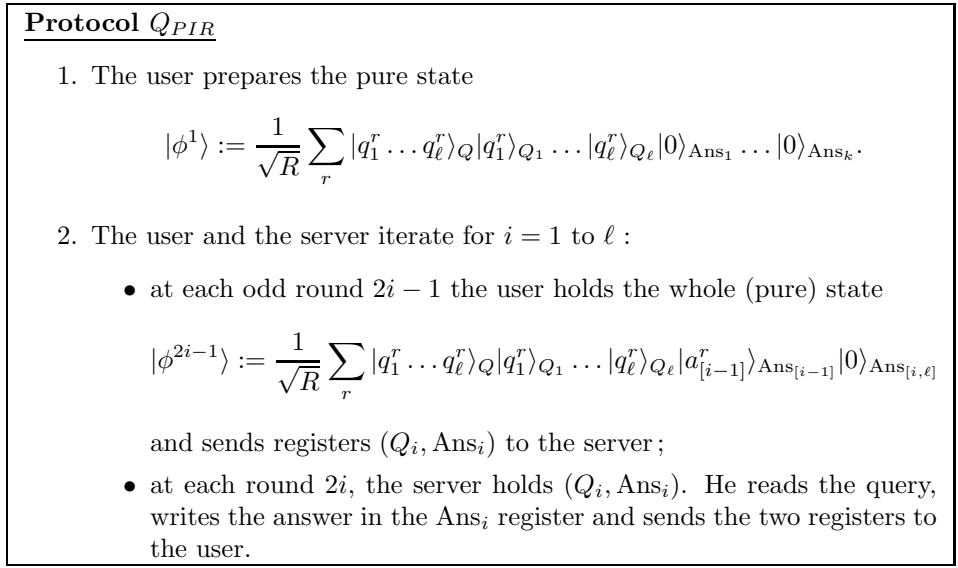

\begin{center}
\fbox{
\begin{minipage}{12 cm} 
\underline{\textbf{Protocol $Q_{PIR}$}}
\begin{enumerate}
	\item The user prepares the pure state
	\[
		\ket{\phi^1} := \frac 1 {\sqrt{R}} \sum_r \ket{q_1^r\dots q_\ell^r}_Q\ket {q_1^r}_{Q_1}\dots \ket {q_\ell^r}_{Q_\ell}\ket {0}_{\Ans_1}  \dots \ket {0}_{\Ans_k}. 
	\]
	\item The user and the server iterate for $i=1$ to $\ell$ :
	\begin{itemize}
		\item at each odd round $2i-1$ the user holds the whole (pure) state
		\begin{equation*}
			\ket{\phi^{2i-1}} := \frac 1 {\sqrt{R}} \sum_r \ket{q_1^r\dots q_\ell^r}_Q\ket {q_1^r}_{Q_1}\dots \ket {q_\ell^r}_{Q_\ell}\ket {a_{[i-1]}^r}_{\Ans_{[i-1]}}\ket {0}_{\Ans_{[i,\ell]}} 
		\end{equation*}
		and sends registers $(Q_i,\Ans_i)$ to the server\,;
		\item at each round $2i$, the server holds $(Q_i,\Ans_i)$. He reads the query, writes the answer in the $\Ans_i$ register and sends the two registers to the user. 
	\end{itemize}
\end{enumerate}
\end{minipage}
}
\end{center}
\caption{Quantum protocol simulating $\Pi_{PIR}$ with one server.}\label{fig:algorithmPIR}
\end{figure}

This protocol indeed simulates the classical protocol $\Pi_{PIR}$: at the end of the protocol, the user holds
\begin{equation*}
	\ket{\phi^{2k}} = \frac{1}{\sqrt{R}}\sum_r  \ket{q_1^r\dots q_\ell^r}_Q\ket {q_1^r}_{Q_1}\dots \ket {q_\ell^r}_{Q_\ell} \ket {a_{[\ell]}^r}_{\Ans_{[\ell]}}
\end{equation*}
and by measuring in the computational basis, he gets a uniformly random $\ell$-tuple of queries and their answers, hence he has the same success probability as the user in the classical scheme. The communication complexity of Protocol $Q_{PIR}$ is $2\ell(m_a+m_q)$ qubits. We now describe its privacy.

\begin{theorem}\label{PIR}
For the above protocol $\Pi_{PIR}$ under uniform distribution of inputs, we have 
	\begin{eqnarray*}
		L_{S}(\Pi_{PIR})  =  O(\ell (m_a + m_q)) \quad &,& \quad L_U(\Pi_{PIR}) = 0, \\
		SIC_S(\Pi_{PIR}) = O(\ell (m_a + m_q) )\quad &, &\quad SIC_U(\Pi_{PIR}) = \Omega(\log(n)) , \\
		QIC_S(\Pi_{PIR}) = O(\ell (m_a + m_q) )\quad &, & \quad QIC_U(\Pi_{PIR}) = \Omega(\log(n)) , 
	\end{eqnarray*}
where the queries and answers are in $\{0,1\}^{m_q}$ and $\{0,1\}^{m_a}$.
\end{theorem}
\begin{proof}
The first statement is obvious since in $\Pi_{PIR}$ the total communication is $2\ell (m_a + m_q)$ and 
hence $L_{S}(\Pi_{PIR}), SIC_{S}(\Pi_{PIR}), QIC_{S}(\Pi_{PIR})$ are $O(\ell (m_a + m_q))$.

As for the privacy of the user, note that each message independently does not leak any information about the user's input, since the quantum message is exactly the same distribution over classical queries that each server receives in the classical scheme, which we know is perfectly private. Moreover, in an honest execution, the server does not keep anything in his workspace, since otherwise the state Alice has in the followinf round will not be the prescribed pure state, and hence the privacy loss is 0.


Now, for $SIC_U(\Pi_{PIR})$ (and  hence for $QIC_U(\Pi_{PIR})$) we know from Theorem 3.2 in \cite{JRS09} that, when the parties are allowed to run the protocol with superpositions of their inputs, if the user leaks at most $b$ bits about his input, then the server has to leak at least $\Omega(n/2^{O(b)})$ bits about his database, or equivalently, if the server leaks at most $t$ bits about the database, then the user must leak at least $\Omega (\log (n/t))$ bits about his input. Since in our scheme the communication is bounded by polylog($n$), we obtain that the user has to leak at least $\Omega(\log(n))$ about his input.
\end{proof}

\paragraph*{Application: a quantum protocol for PIR with polylogarithmic privacy loss}

We consider the classical scheme proposed in \cite{CGKS95P,CGKS95}.

\begin{lemma}[See Corollary 4 in \cite{CGKS95P}]\label{lemma-CGKS95P}
	There are (classical) private information retrieval schemes for $\frac 12 \cdot(\log_2 n + \log_2 \log_2 n) +1$ servers, each holding $n$ bits of data, so that the communication complexity is $\frac 12 \cdot(1+o(1))\cdot\log_2^2 n \cdot\log_2 \log_2(2n)$.
\end{lemma}
By converting the classical protocol of Lemma \ref{lemma-CGKS95P} into a one-server quantum protocol by the above construction, and applying 
Theorem \ref{PIR}, we obtain the following result.

\begin{corollary}
		There exists a one-server quantum protocol for Private Information Retrieval, with communication complexity $O(\log^2(n) \cdot \log\log(n))$, such that: 
\begin{itemize}
	\item the user leaks no information ;
	\item the server leaks $O\big(\mathrm{polylog}(n)\big)$ information.
\end{itemize}
\end{corollary}

\section{Logarithmic scheme for PIR with prior entanglement}

We now study one-server quantum private information retrieval in the same setting as in the previous section,
but allowing prior entanglement between the server and the user, and construct a protocol with privacy loss and communication 
complexity $O(\log (n))$. For simplicity we will assume in this section that $n=2^\ell$, and write the user's input using its binary representation as $i=i_1i_2\ldots i_\ell$, where $i_1,\ldots,i_\ell$ are bits such that $i=1+\sum_{k=1}^\ell i_k 2^{\ell-k}$. The case where $n$ is not a power of two can be dealt in a similar way, or simply by adding zeros to the database in order to obtain a size that is a power of two.  

For convenience we introduce the following notation.
\begin{definition}\label{def}
Let $s$ be any positive integer, and $z$ be any binary string of length $2^s$. 
Define $z[0]$ and $z[1]$ as the first and second halves of the string $z$, respectively. For
any $k\in\{2,\ldots,s\}$ and any $k$ bits $j_1,\ldots,j_k$, let $z[j_1,\ldots,j_k]$ be the binary string of 
length $2^{s-k}$ defined by the recurrence relation
$
z[j_1,\ldots,j_k]=(z[j_1,\ldots,j_{k-1}]) [j_k].
$
 \end{definition}

Let us consider an example to illustrate this definition:
if $s=3$ and $z=10100110$, then  $z[0]=1010$, $z[1]=0110$,
$z[0,0]=10$, $z[0,1]=10$, $z[1,0]=01$, $z[1,1]=10$ and, for instance, 
$z[0,0,0]=1$ or $z[1,0,1]=0$. 
Note that, with these definitions, the bit $x_i$ that the user wants to output in a protocol for 
Private Information Retrieval is $x[i_1,\ldots,i_\ell]$.

Our protocol will use, besides the two registers containing the inputs, 
the following quantum registers:
\begin{itemize}
\item
$\ell$ quantum registers $R_1,\ldots,R_\ell$ where
$R_k$ is a register of $2^{\ell-k}$ qubits for $k\in \{1,\ldots,\ell\}$;
\item
$\ell$ quantum registers $R'_1,\ldots,R'_\ell$ where
$R'_k$ is a register of $2^{\ell-k}$ qubits for $k\in \{1,\ldots,\ell\}$;
\item 
two one-qubit quantum registers $Q_{0}$ and $Q_{1}$. 
\end{itemize}

Define the unitary operator $V_1$
acting on $(R_{1},Q_{0},Q_{1})$
as follows:
\[
V_1\left(
\ket{z}_{R_1}\ket{a}_{Q_0}\ket{b}_{Q_1}
\right)
=
\ket{z}_{R_1}\ket{a\oplus z\cdot x[0]}_{Q_0}\ket{b\oplus z\cdot x[1]}_{Q_1}
\]
for any string $z\in \{0,1\}^{2^{\ell-1}}$ and any bits $a,b\in\{0,1\}$.
For any integer $k\in \{2,\ldots,\ell\}$, we define the unitary operator $V_k$
acting on $(R_{k-1},R_{k},Q_{0},Q_{1})$
as follows:
\[
V_k\left(
\ket{y}_{R_{k-1}}\ket{z}_{R_k}\ket{a}_{Q_0}\ket{b}_{Q_1}
\right)
=
\ket{y}_{R_{k-1}}\ket{z}_{R_k}\ket{a\oplus z\cdot y[0]}_{Q_0}\ket{b\oplus z\cdot y[1]}_{Q_1}
\]
for any strings $y\in \{0,1\}^{2^{\ell-k+1}}$, $z\in \{0,1\}^{2^{\ell-k}}$ and any bits $a,b\in\{0,1\}$.

For any integer $k\in \{1,\ldots,\ell\}$, 
define the state 
\begin{equation*}
\ket{\Phi_k}_{(R_k,R'_k)}=\frac{1}{\sqrt{2^{2^{\ell-k}}}}\sum_{z\in \{0,1\}^{2^{\ell-k}}}\ket{z}_{R_k}\ket{z}_{R'_k}.
\end{equation*}
We assume that the 
server and the user initially share the quantum state
\[
\ket{\Phi_1}_{(R_1,R'_1)}\otimes\ket{\Phi_{2}}_{(R_2,R'_2)}\otimes\cdots\otimes\ket{\Phi_{\ell}}_{(R_\ell,R'_\ell)}\otimes \ket{0}_{Q_0}\ket{0}_{Q_1},
\]
where $R_1,\ldots,R_\ell,{Q_0},{Q_1}$ are owned by the server and ${R'_1},\ldots,{R'_\ell}$ are owned by the user.
Our quantum protocol is given in Fig.~\ref{fig:algorithmA}.

\begin{figure}[ht!]
\begin{center}
\fbox{
\begin{minipage}{12 cm} 
\underline{\textbf{Protocol $\mathcal{P}_{PIR}$}}
\begin{enumerate}
\item
For $k$ from 1 to $\ell$, the server and the user do the following:
\begin{enumerate}
\item
The server applies $V_k$,
and then sends Registers ${Q_0}$ and ${Q_1}$ to the user;
\item
The user applies the Pauli gate ${Z}$ over Register ${Q_{i_k}}$ and sends back Registers ${Q_0}$ and ${Q_1}$ to the server;
\item
The server applies $V_k$, and applies a Hadamard transform on each of the $2^{\ell-k}$ qubits in Register ${R_k}$;
\item
The user applies a Hadamard transform on each of the $2^{\ell-k}$ qubits in  Register ${R'_k}$.
\end{enumerate} 
\item
The server sends Register ${R_\ell}$ to the user.
\end{enumerate} 
\end{minipage}
}
\end{center}
\caption{Quantum protocol for private information retrieval with prior entanglement.}\label{fig:algorithmA}
\end{figure}

We analyze the correctness, the complexity and 
the privacy of Protocol $\mathcal{P}_{PIR}$ in the Appendix,
and prove the following theorem.
\begin{theorem}\label{QPIR}
The protocol $\mathcal{P}_{PIR}$ for input size $n=2^\ell$ has communication complexity $4\log(n)+1$ qubits and
correctly computes the index function. Moreover, 
under uniform distribution of inputs, we have 
	\begin{eqnarray*}
		L_S(\mathcal{P}_{PIR})  \le 2\log(n)+1 \quad &,& \quad L_U(\mathcal{P}_{PIR}) = 0, \\
		SIC_S(\mathcal{P}_{PIR}) \le 2\log(n)+1\quad &, &SIC_U(\mathcal{P}_{PIR}) = \Omega(\log(n)), \\
		QIC_S(\mathcal{P}_{PIR}) \le 2\log(n)+1\quad &, & QIC_U(\mathcal{P}_{PIR}) = \Omega(\log(n)). 
	\end{eqnarray*}
\end{theorem}

The proof relies on the following lemma, which can be easily shown by recursion on $k$. 

\begin{lemma}\label{lemma-pe}
Assume that Protocol $\mathcal{P}_{PIR}$ is applied when the server's input is $x\in\{0,1\}^{2^{\ell}}$ and the user's input is $i\in\{0,1\}^{\ell}$. Then, 
at the end of the $k$-th iteration of the loop in Step~1, the state of the quantum system is (omitting
a global normalization factor)
\begin{align*}
\hspace{-6mm}
\left[
\sum_{y^1\!,\ldots,y^k}\!\!\!\Big(
\ket{y^1}_{{R_{1}}}\ket{x[i_1]\oplus y^1}_{{R'_{1}}}\!\!\otimes
\bigotimes_{j=2}^k
\ket{y^j}_{{R_{j}}}\ket{y^{j-1}[i_j]\oplus y^j}_{{R'_{j}}}
\Big)\!
\right]\!\!
\otimes\!
\ket{0}_{{Q_0}}\!\ket{0}_{{Q_1}}\!
\otimes\!\!
\left[
\bigotimes_{j=k+1}^\ell
\ket{\Phi_{j}}_{({R_{j}},{R'_{j}})}
\right]\!\!,
\end{align*}
where the sum is over all strings 
$y^1\in\{0,1\}^{2^{\ell-1}},\ldots, y^k\in\{0,1\}^{2^{\ell-k}}$.
\end{lemma}

\begin{proof}[Proof of Theorem \ref{QPIR}]
Since each iteration of the loop in Step 1 uses four qubits of communication,
and one additional qubit is used at Step 2, the overall communication complexity 
is $4\ell+1$.

Next, we show that this protocol correctly computes the index function, i.e., 
the user can output $x_i$. 
From Lemma \ref{lemma-pe},  
the state of the quantum system 
at the end of Protocol~$\mathcal{P}_{PIR}$
is (omitting a global normalization factor)
\[
\sum_{y^1,\ldots,y^\ell}\!\!\! \ket{y^1}_{{R_{1}}}\ket{x[i_1]\oplus y^1}_{{R'_{1}}}
\ket{y^2}_{{R_{2}}}\ket{y^1[i_2]\oplus y^2}_{{R'_{2}}}
\cdots
\ket{y^\ell}_{{R_{\ell}}}\ket{y^{\ell-1}[i_\ell]\oplus y^\ell}_{{R'_{\ell}}}
\ket{0}_{{Q_0}}\ket{0}_{{Q_1}},
\]
where the server owns Registers
${R_{1}},\ldots, {R_{\ell-1}},{Q_0},{Q_1}$, 
and the user owns Register 
${R_{\ell}}$ and Registers ${R'_{1}},\ldots,{R'_{\ell}}$.
If the server and the user measure all their registers, the user obtains strings
$a^\ell,b^1,\ldots,b^\ell$ such that
\[
\left\{\begin{array}{lll}
a^\ell&=&y^\ell,\\
b^1&=&x[i_1]\oplus y^1,\\
b^2&=&y^1[i_2]\oplus y^2,\\
\,\,\vdots&\,\vdots&\,\,\,\vdots\\
b^{\ell}&=&y^{\ell-1}[i_\ell]\oplus y^\ell,
\end{array}
\right.
\]
for some strings $y^1,\ldots,y^\ell$ corresponding to the server's measurement outcomes. 
Note that 
\[
x[i_1,i_2,\ldots,i_\ell]=
b^1[i_2,\ldots,i_\ell]\oplus b^2[i_3,\ldots,i_\ell]\oplus\cdots\oplus b^{\ell-1}[i_\ell]\oplus b^\ell\oplus a^\ell,
\]
which means that the user can recover $x_i=x[i_1,i_2,\ldots,i_\ell]$ from his measurement outcomes.

The upper bounds on $L_S(\mathcal{P}_{PIR})$, $SIC_S(\mathcal{P}_{PIR})$ and
$QIC_S(\mathcal{P}_{PIR})$ follow from the observation that 
the total length of the messages received by the user is $2\ell+1$.
The lower bounds on $SIC_U(\mathcal{P}_{PIR})$ and  
$QIC_U(\mathcal{P}_{PIR})$ follow from the same argument 
(based on \cite{JRS09}) as in Theorem \ref{PIR}. 

Finally, let us prove that $L_U(\mathcal{P}_{PIR}) = 0$, by showing that the 
server's state  just after receiving the message from the user during the $k$-th iteration of Step 1 of Protocol $\mathcal{P}_{PIR}$ is independent of $i$, for 
each $k\in\{1,\ldots,\ell\}$. In the case $k=1$, the state of the registers owned by 
the server just after receiving the message from the user is, omitting a global normalization factor, 
\[
\left[
\sum_{z\in\{0,1\}^{{2^{\ell-1}}}}\!\!\!
\ket{\Psi(z)}\bra{\Psi(z)}
\right]
\otimes
\left[
\bigotimes_{j=2}^\ell
\sum_{z\in\{0,1\}^{2^{\ell-j}}}\!\!
\ket{z}_{{R_{j}}}
\bra{z}_{{R_{j}}}
\right]\!,
\]
where 
$
\ket{\Psi(z)}=
(-1)^{x[i_1]\cdot z}
\ket{z}_{{R_{1}}}
\ket{x[0]\cdot z}_{{Q_0}}\ket{x[1]\cdot z}_{{Q_1}}.
$
Since 
\[
\ket{\Psi(z)}\bra{\Psi(z)}
=\ket{z}_{{R_{1}}}
\ket{x[0]\cdot z}_{{Q_0}}\ket{x[1]\cdot z}_{{Q_1}}
\bra{z}_{{R_{1}}}
\bra{x[0]\cdot z}_{{Q_0}}\bra{x[1]\cdot z}_{{Q_1}}
\]
is independent of $i$, the above state is also independent of $i$.
For the case $k\ge 2$, by using Lemma \ref{lemma-pe}, 
the state of the registers owned by 
the server just after receiving the $k$-th message from the user is, omitting a global normalization factor,
\[
\left[
\sum_{y^1,\ldots,y^{k-1}}\!
\sum_{z\in\{0,1\}^{{2^{\ell-k}}}}\!\!\!\!
\ket{\Psi_x(y^1,\ldots,y^{k-1},z)}\bra{\Psi_x(y^1,\ldots,y^{k-1},z)}
\right]
\otimes
\left[
\bigotimes_{j=k+1}^\ell
\sum_{z\in\{0,1\}^{2^{\ell-j}}}\!\!\!\!\!
\ket{z}_{{R_{j}}}
\bra{z}_{{R_{j}}}
\right]\!\!,
\]
where 
\[
\ket{\Psi_x(y^1,\ldots,y^{k-1},z)}=
(-1)^{y^{k-1}[i_k]\cdot z}
\ket{y^1}_{{R_{1}}}\cdots\ket{y^{k-1}}_{{R_{k-1}}}
\ket{z}_{{R_{k}}}
\ket{y^{k-1}[0]\cdot z}_{{Q_0}}\ket{y^{k-1}[1]\cdot z}_{{Q_1}},
\]
and is again independent of $i$.
\end{proof}

\section{Acknowledgments}
The authors are grateful to Ronald de Wolf for helpful comments about this work, and for pointing out Ref.~\cite{RB13}.
The authors are also grateful to Rahul Jain for helpful discussion.
Iordanis Kerenidis and Mathieu Lauri{\`{e}}re have been supported by the ERC
grant QCC and the EU grant QAlgo.
Fran{\c c}ois Le Gall has been supported by the
Grant-in-Aid for Scientific Research~(A)~No.~24240001 of the Japan Society for the Promotion of Science
and the Grant-in-Aid for Scientific Research on Innovative Areas~No.~24106009 of
the Ministry of Education, Culture, Sports, Science and Technology in Japan.

\hspace{1cm}

\end{document}